\documentclass[11pt]{article}

\usepackage[left=1in, top=1in, right=1in, bottom=1in]{geometry}
\usepackage{amsmath}
\usepackage{amssymb}
\usepackage{amsthm}
\usepackage{paralist}
\usepackage{algorithmicx}
\usepackage[noend]{algpseudocode}
\usepackage[linesnumbered,lined,boxed,commentsnumbered]{algorithm2e}
\usepackage{color}
\usepackage[numbers]{natbib}
\usepackage[pdfpagelabels,pdfpagemode=None]{hyperref}
\usepackage{aliascnt}

\usepackage{mathtools}
\usepackage{bm}

\newtheorem{theorem}{Theorem}%

\newaliascnt{lemma}{theorem}
\newtheorem{lemma}[lemma]{Lemma}%
\aliascntresetthe{lemma}

\newtheorem{claim}{Claim}%

\newtheorem{corollary}{Corollary}%

\theoremstyle{definition}
\newtheorem{definition}{Definition}

\newtheorem{example}{Example}

\newtheorem{remark}{Remark}

\newtheorem{fact}{Fact}

\renewcommand\footnotemark{}

\newcommand{\AutoAdjust}[3]{\mathchoice{ \left #1 #2  \right #3}{#1 #2 #3}{#1 #2 #3}{#1 #2 #3} }
\newcommand{\Xcomment}[1]{{}}

\newcommand{\InBrackets}[1]{\AutoAdjust{[}{#1}{]}}
\newcommand{\Ex}[2][]{\operatorname{\mathbf E}_{#1}\InBrackets{#2}}
\newcommand{\Prx}[2][]{\operatorname{\mathbf{Pr}}_{#1}\InBrackets{#2}}

\newcommand{\dd}{\mathrm{d}}  
\newcommand{\given}{\;\mid\;}
\newcommand{\eps}{\epsilon}

\newcommand{\noaccents}[1]{#1}

\newcommand{\R}{\mathbb{R}}
\newcommand{\cC}{\mathcal{C}}
\newcommand{\cD}{\mathcal{D}}
\newcommand{\cF}{\mathcal{F}}
\newcommand{\cL}{\mathcal{L}}

\newcommand{\newagentvar}[3][\noaccents]{%
\expandafter\newcommand\expandafter{\csname #2\endcsname}{#1{#3}}%
\expandafter\newcommand\expandafter{\csname #2s\endcsname}{#1{\boldsymbol{#3}}}%
\expandafter\newcommand\expandafter{\csname #2smi\endcsname}[1][i]{#1{\boldsymbol{#3}}_{-##1}}%
\expandafter\newcommand\expandafter{\csname #2i\endcsname}[1][i]{#1{#3}_{##1}}%
\expandafter\newcommand\expandafter{\csname #2ith\endcsname}[1][i]{#1{#3}_{(##1)}}%
}

\newcommand{\typedist}{\cD}

\newcommand{\sig}{s}

\newagentvar{type}{t}
\newagentvar{alloc}{x}
\newagentvar{pay}{p}
\newagentvar{typespace}{T}
\newagentvar{sigspace}{S}

\newcommand{\virt}{\varphi}

\newcommand{\highval}{h}

\newagentvar{allocopt}{x^*}
\newagentvar{payopt}{p^*}

\DeclareMathOperator{\secmax}{secmax}

\newagentvar{val}{v}
\newagentvar{util}{u}
\newagentvar{distmat}{\dist}
\newagentvar{payment}{p}
\DeclareMathOperator{\DRev}{DRev}
\DeclareMathOperator{\Rev}{Rev}

\newcommand{\bfmu}{\mathbf{\mu}}
\newcommand{\bflambda}{\mathbf{\lambda}}

\title{The Value of Information Concealment \footnote{The authors would like to thank Shuchi Chawla and Anna Karlin for helpful discussions during the early stage of this work.}}
\date{}

\author{
	Hu Fu \\
	University of British Columbia \\
	Department of Computer Science \\
	\tt{hufu@cs.ubc.ca}
	\and
	Chris Liaw \\
	University of British Columbia \\
	Department of Computer Science \\
	\tt{cvliaw@cs.ubc.ca}
	\and
	Pinyan Lu \\
	Shanghai University of Finance and Economics \\
	School of Information Management and Engineering \\
	\tt{lu.pinyan@mail.shufe.edu.cn}
	\and
	Zhihao Gavin Tang \\
	University of Hong Kong \\
	Department of Computer Science \\
	\tt{zhtang@cs.hku.hk}
}

\begin{document}

\begin{titlepage}
\maketitle

\begin{abstract}
We consider a revenue optimizing seller selling a single item to a buyer, on whose private value the seller has a noisy signal.  We show that, when the signal is kept private, arbitrarily more revenue could potentially be extracted than if the signal is leaked or revealed.  We then show that, if the seller is not allowed to make payments to the buyer, the gap between the two is bounded by a multiplicative factor of~$3$ when the value distribution conditioning on each signal is regular.  We give examples showing that both conditions are necessary for a constant bound to hold.

We connect this scenario to multi-bidder single-item auctions where bidders' values are correlated.  Similarly to the setting above, we show that the revenue of a Bayesian incentive compatible, ex post individually rational auction can be arbitrarily larger than that of a dominant strategy incentive compatible auction, whereas the two are no more than a factor of $5$ apart if the auctioneer never pays the bidders and if each bidder's value conditioning on the others' is drawn according to a regular distribution.  The upper bounds in both settings degrade gracefully when the distribution is a mixture of a small number of regular distributions.
\end{abstract}

\thispagestyle{empty}
\end{titlepage}

\section{Introduction}
\label{sec:intro}


Revenue maximization when selling a single item to a single buyer is a fundamental and well-understood problem in mechanism design.  In the classical model, the buyer holds a private value~$\val$ for the item, and the seller only knows the distribution~$\typedist$ from which $\val$ is drawn.  The revenue-maximizing strategy of the seller is to offer the buyer an optimal take-it-or-leave-it price.  This simple fact is one of the earliest insights offered by the revenue optimization literature \citep{M81, RZ83}, and also one of the key building blocks of many recent approximation results in more complicated scenarios \citep[e.g.][]{CHMS10, BILW16}.

Implicit in this classical model is a simple information structure: the buyer is completely informed, and the seller only partially so.  As a consequence,  when taking a price, the buyer usually derives a positive utility knows as the \emph{information rent} due to this information asymmetry.  In modern online markets, however, there is often information asymmetry in the other direction as well.  Even if the buyer always knows fully her value (which may or may not be the case in practice), the seller may still possess partial information of which the buyer cannot be sure.  For example, transactions in the near past with other buyers may suggest market trends which inform the seller but are not precisely known to the buyer.  The notorious obscurity of machine learning techonologies also contribute to such asymmetry --- whereas the seller may have more or less accurate estimates on the buyer's value based on various observable attributes, the buyer may not be able to determine which attributes of hers are picked up as cues by the seller's learning algorithm, and hence could not determine what estimate the seller has.  In these scenarios, the buyer is uncertain about the information held by the seller, even though she is completely informed of her value, seemingly the central information in the transaction.

This paper studies the implication of such information asymmetry on the seller's revenue.  The basic question we raise is: \emph{Are there nontrivial ways for the seller to exploit this asymmetry so as to increase his revenue?}  The answer, in short, is that sophisticated exploitation is possible and in some cases could drive up the revenue unboundedly compared with ``na\"ive'' uses of this information asymmetry; however, under a few conditions that seem naturally satisfied in practical scenarios, the benefit of such sophistication is bounded by a small muliplicative constant factor.  Below we explain what we mean by ``sophisticated'' and ``na\"ive'' pricings, after we lay down the basic model.


\paragraph{The Model.}

We model the simplest information asymmetry by having a value~$\val$ and a signal~$\sig$ drawn from a commonly known correlated distribution; the buyer only sees her value~$\val$ for the item and the seller only sees~$\sig$.  The seller aims to optimize his revenue, taking expectation over the joint distribution.

\emph{Na\"ive pricing:} Upon observing~$\sig$, the seller has a distribution over~$\val$ conditioning on~$\sig$.  We say the seller na\"ively uses the signal if he simply adopts the optimal strategy for each such conditional distribution --- as we saw above, this amounts to posting a take-it-or-leave-it price optimally for each~$\sig$.  It is easy to see that this is also the best strategy of the seller if $\sig$ is publicized or leaked.  In an e-commerce scenario, this could be the case where the seller publicizes his learning algorithm or whichever information gathered pertaining to the buyer.

It may come as a surprise that the na\"ive pricing may not be optimal; after all, in a sense, it optimizes the revenue pointwise.  The suboptimality of this approach lies precisely in its failure to exploit information asymmetry.  Publicizing the signal reduces its worth to a mere sharper estimte of the value distribution, and gives up its uncertainty, when kept private, in the buyer's eyes.  We show in \autoref{sec:examples} that \emph{this uncertainty can be leveraged to elicit more information from the buyer and hence increase the seller's revenue, in some cases by an arbitrarily large multiplicative factor.}


\paragraph{Main result.}

Our upper bound on the worth of concealing the signal is the most technically-sophisticated part of the paper.  Under two conditions, we show that the optimal revenue with private signal is no more than $3$ times the revenue of optimal posted price under each (public) signal.  The first condition asks that the seller never pay money to the buyer, which seems a natural requirement in many scenarios.  The second condition requires that each conditional distributions of the value~$\val$ is \emph{regular} in \citet{M81}'s sense.\footnote{Regularity is a standard notion in the literature of revenue optimization, to be precisely defined in \autoref{sec:prelim}.}  Viewed alternatively, our bound suggests that a seller looking to boost the value of private signals may consider making payments to the buyer occassionally or sharpening the signals in such a way that the resulting conditional value distributions have interleaved supports (see \autoref{ex:weird-dist} for an instance of this).  

For the readers familiar with the revenue optimization literature, we remark that regularity of value distributions usually does not mitigate the gaps between revenues of different mechanisms, not least because the most common distribution used in showing such large gaps, the equal revenue distribution, is itself regular.  As we explain later, the regularity condition enters our analysis in a non-standard way.  

Regularity is often deemed a natural condition when the buyer can be identified as a certain category of customer, whereas uncertainty on the category results in irregular distributions that is a mixture of regular distributions \citep[see, e.g.,][]{SS13}.  Our upper bound degrades gracefully for mixtures of a small number of regular distributions.



\paragraph{Connection to multi-bidder auctions.}

We are not the first to observe that correlated external information can be used to extract a buyer's information rent.  For multi-bidder auctions where the bidders' values are drawn from correlated distributions, classical works by \citet{CM85, CM88} showed ways to extract from each bidder her full value by leveraging her competitors' bids.   The similarity between the two settings is immediate when one takes the competitors' bids as private signals.  However, \citeauthor{CM85}'s results do not apply to us because we constrain our seller's mechanism much more stringently:  crucial for \citeauthor{CM85}'s mechanisms is that the auctioneer may charge the buyer anything as long as the buyer, conditioning on any value she holds, has nonnegative \emph{expected} utility.  This is known as \emph{interim} individual rationality (IR).  In contrast, we insist on the more realistic \emph{ex post} IR constraint, which stipulates that the buyer should never have a negative utility no matter what signal is realized.  This key difference is what makes our examples of unbounded gaps nontrivial and our constant upper bound possible.

In multi-bidder auctions,  the solution concept corresponding to na\"ive pricing or pricing with publicized signals is \emph{dominant strategy incentive compatibility} (DSIC), whereas the more sophisticated selling with private signals corresponds to \emph{Bayesian incentive compatible} (BIC) auctions.  The second main contribution of this paper is the first study of the revenues of DSIC and BIC single-item auctions, which complements several recent works that compare the revenues of DSIC and BIC \emph{multi-item} auctions \citep{TW16, Yao17}.  (Under the looser interim IR constraint, the revenue of DSIC and BIC auctions was almost fully characterized by \citet{CM85, CM88}.)  

We show unbounded gaps between the revenues of DSIC and BIC auctions (under ex post individual rationality); with non-negative payments and regular distributions, we show that \citet{R01}'s DSIC lookahead auction 5-approximates the optimal BIC revenue.  The first part follows immediately from the single-buyer examples, whereas the 5-approximation is technically more nontrivial, as we discuss below.


\paragraph{Our techniques.}

The driving horse for our main result is the duality framework which was recently promoted by \citet{CDW16}  and subsequently applied to various settings \citep[e.g.][]{FGKK16, DW17, DHP17}.  The basic approach is to write the optimal revenue of a target mechanism (in our case, the revenue with private signals, or the optimal BIC revenue) as the objective of a linear program and then Lagrangify the IC and IR constraints;  the value of the resulting Lagrangian with any dual variables serves as an upper bound on the optimal revenue and can be used as a benchmark for approximation.  What is peculiar about the argument in this work is that, since the the signals' uncertainty is key to revenue maximization, the dual variables controlling the buyer's incentives cannot be engineered for each value distribution conditioned on a signal.  That is, these dual variables cannot be functions of the signals.  On the other hand, standard dual variables giving rise to the public signal revenue via \citeauthor{M81}'s \emph{virtual values} are distribution dependent.  In general, distribution-independent dual variables do not produce quantities that can be bounded by the virtual values.  We overcome this difficulty by constructing dual variables which, besides the ``virtual value terms'', generate another term that grows in the worse case with the welfare.  A fairly non-standard argument shows that regularity brings this additional term down to comparable with na\"ive pricing revenues (or DSIC revenue in multi-bidder auctions).

Multi-bidder auctions pose another difficulty: the seller not only maximizes the revenue from each bidder, but is also subject to the item's availability: two bidders cannot be both assigned the item.  Revenue optimal auctions are already complicated even under the more stringent DSIC constraint \citep[e.g.\@][]{DFK15, PP15, CKK16}.  One way to disentangle multiple bidders' interaction is provided by \citet{R01}'s DSIC \emph{lookahead auction}, which maximizes its revenue from the highest bidders and gives up lower bidders.  \citeauthor{R01} showed that this 2-approximates the optimal DSIC revenue.  However, the natural extension of \citeauthor{R01}'s auction to the BIC auctions fails, due to a subtle reason we explain briefly at the end of the paper (\autoref{sec:conclusion}).  Instead, we take the duality approach again, and appeal to a property possessed by the particular dual variables we constructed for the single buyer case (\autoref{cl:dual-bound} in \autoref{sec:lookahead}).  This allows us to bound the lower bidders' contribution in the Langrangian by twice the second price auction revenue, and the problem boils down to maximizing revenue from the highest bidder, which is a single buyer problem, where we could apply our previous argument.  We note that bounding lower bidders' contribution to revenue by the second highest value is the essence of \citet{R01}'s original proof, and the idea that this can be done in the Langragian was implicit in several duality-based works \citep[e.g.][]{CDW16, EFFTW17}.

\section{Preliminaries}
\label{sec:prelim}

\paragraph{Single-Buyer Pricing with Signals.}


A seller tries to sell a single item to a buyer, who has a private type $\type$, which indicates the buyer's value of the item; throughout the paper we use types $\type$ and values $\val$ interchangeably.  
The seller observes a private signal $\sig \in \sigspace$.  The pair $(\type, \sig)$ is drawn from a commonly known joing distribution $\cD$, whose density function we denote by $f(\cdot, \cdot)$.  We also write 
$F_s(t) \coloneqq \int_{t' \leq t} f(t',s) \: \dd t'$.
A selling mechanism consists of an allocation rule $\alloc(\val, \sig)$, indicating the probability with which the buyer gets the item when she reports her value as~$\val$ and the seller sees signal~$\sig$, and a payment rule $\payment(\val, \sig)$, the payment made by the buyer to the seller.  We say the mechanism has \emph{no negative payment} if $\payment$ is never negative, i.e., the seller never pays the buyer.  The expected revenue of a mechanism is $\Ex[(\val, \sig) \sim \cD]{\payment(\val, \sig)}$.

In \emph{na\"ive pricing}, or \emph{pricing with publicized or leaked signals}, the seller adopts, for each observed~$\sig$, a mechanism for the value distribution conditioned on~$\sig$.  Such a mechanism needs to be (i) ex post individually rational (IR): for any $\val$ and $\sig$, $\val \alloc(\val, \sig) - \payment(\val, \sig) \geq 0$, and (ii) incentive compatible (IC): for any $\val$, $\val'$ and $\sig$, $\val \cdot \alloc(\val, \sig) - \payment(\val, \sig) \geq \val \cdot \alloc(\val', \sig) - \payment(\val', \sig)$.  

\begin{definition}[Virtual Values]
	Given a value distribution with density function~$f$ and cumulative density function~$F$, the \emph{Myerson virtual value} is $\virt(\val) \coloneqq \val - \frac{1 - F(\val)}{f(\val)}$.  Given a joint distribution $\cD$ on $(\val, \sig)$, we denote by $\virt_\sig(\val)$ the virtual value of $\val$ in the conditional distribution $F_\sig$ given $\sig$.
\end{definition}

\begin{lemma}[\citeauthor{M81}\citeyear{M81}]
\label{lem:Myerson}
A mechanism is IC only if its allocation rule~$\alloc$ is monotone non-decreasing, and its expected revenue for a value distribution~$F$ (without signals) is the \emph{virtual surplus} $\int_{\val} \virt(\val) \alloc(\val) \: \dd F(\val)$.  The revenue-optimal mechanism for any~$F$ posts a take-it-or-leave-it price.
\end{lemma}

\autoref{lem:Myerson} implies that with public signals, posting a price optimal for each signal~$\sig$ maximizes the seller's revenue.

For mechanisms with \emph{private signals}, besides being (i) ex post IR, the IC constraint weakens to: (ii') for all $\val$ and $\val'$, $\Ex[\sig \sim \cD(\sig | \val)]{\val \cdot \alloc(\val, \sig) - \payment(\val, \sig)} \geq \Ex[\sig \sim \cD(\sig | \val)]{\val \cdot \alloc(\val', \sig) - \payment(\val', \sig)}$, where $\cD(\sig | \val)$ denotes the conditional distribution of $\sig$ given~$\val$.  Equivalently, the constraint can be written as $\int_{\sig} (\val \cdot \alloc(\val, \sig) - \payment(\val, \sig)) f(\val, \sig) \: \dd \sig \geq \int_{\sig} (\val \cdot \alloc(\val', \sig) - \payment(\val', \sig)) f(\val, \sig) \: \dd \sig$.

\begin{definition}[(Jointly) Regular Distributions]
	A value distribution $F$ (without signals) is \emph{regular} if its support is an interval $[l, h] \subset \mathbb R$, and its virtual value $\virt$ is nondecreasing on $[l, h]$.  A distribution $\cD$ on $(\val, \sig)$ is \emph{jointly regular} if the conditional distribution of~$\val$ given any signal~$\sig$ is regular.
\end{definition}
Note that joint regularity does not require the conditional value distributions to be supported on the same interval given different signals.



\paragraph{Single-item (Multi-bidder) Auctions}
In a single item auction, one indivisible good is to be allocated to at most one of $n$ bidders.
Bidder $i$ has a private type $t_i \in \R$,  and the types are drawn from a possibly correlated, commonly known joint distribution $\cD$.
We write $\types = (t_1, \ldots, t_n)$ to denote the profile of types and $\typesmi = (\typei[1], \ldots, \typei[i-1], \typei[i+1], \ldots, \typei[n])$.
A mechanism (auction) 
is specified by its allocation rules $\alloci \colon \types \mapsto [0,1]$ and payment rules $\paymenti \colon \types \mapsto \R$.
The utility of bidder $i$ is $u_i(\types) = \typei \cdot \alloci(\types) - \paymenti(\types)$.
The allocation rules are required to satisfy the feasbility constraint: $\sum_i \alloci(\types) \leq 1$, for all $\types$ and~$i$.

Similar to the IC constraint under public signals in single buyer pricing, a mechanism is \emph{dominant strategy incentive compatible} (DSIC) if for each bidder $i$, each type profile $\bm{t}$, and any deviation $\typei'$,
\[
   t_i \cdot x_i(\bm{t}) - p_i(\bm{t}) \geq t_i \cdot x_i(t_i', \bm{t}_{-i}) - p_i(t', \bm{t}_{-i}).
\]
Similar to the IC constraint under private signals, an auction is \emph{Bayesian incentive compatible} (BIC) if for each bidder $i$ with type $t_i$ and possible deviation~$\typei'$,
\[
   \Ex[\bm{t}_{-i} \sim \cD_{\bm{t}_{-i} \mid t_i}]{t_i \cdot x_i(\bm{t}) - p(\bm{t})} \geq
   \Ex[\bm{t}_{-i} \sim \cD_{\bm{t}_{-i} \mid t_i}]{t_i \cdot x_i(t_i', \bm{t}_{-i}) - p(t_i', \bm{t}_{-i})},
\]
where $\cD_{\bm{t}_{-i} \mid t_i}$ denotes the distribution of $\typesmi$ conditioned on bidder~$i$'s type being $t_i$.
 A mechanism is \emph{ex-post individually rational} (IR) if $t_i \cdot x_i(\bm{t}) - p_i(\bm{t}) \geq 0$ for any type profile $\bm{t}$
and \emph{interim IR} if $\Ex[\bm{t}_{-i} \sim \cD_{\bm{t}_{-i} \given t_i}]{t_i \cdot x_i(\bm{t}) - p_i(\bm{t})} \geq 0$.

\citet{M81}'s seminal paper showed that, for a product distribution~$\cD$, a DSIC, ex post IR auction maximizes revenue among all BIC, interim IR auction.  In sharp contrast, for correlated distributions, \citet{CM85, CM88} showed that, except for degenerate cases, interim IR auctions can extract full surplus, i.e., $\Ex[\types]{\max_i \typei}$.  Throughout this paper we require the mechanisms to be ex post IR.

\paragraph{The Lookahead Auction.}

\citet{R01}'s lookahead auction reduces multi-bidder DSIC revenue maximization to pricing for a single bidder, while keeping at least half of the optimal revenue.

\begin{definition}[(Dominant Strategy) Lookahead Auction]\label{def:DSLA}
For each bid vector~$\types$, the lookahead auction choose the highest bidder~$i^*$, and offer an optimal take-it-or-leave-it price conditioning on (i) $\typesmi[i^*]$ and (ii) $\typei[i^*]$ being the highest type.  
\end{definition}

\begin{theorem}[\citeauthor{R01}\citeyear{R01}]
The dominant strategy lookahead auction is DSIC and ex post IR, and extracts at least half of the revenue of the optimal DSIC, ex post IR auction.  
\end{theorem}

We explicitly designate \citeauthor{R01}'s auction as dominant strategy lookahead, since there is a natural extension of it for BIC auctions, which we call the Bayesian lookahead auction.  However, we are not able to show general revenue guarantee for it. (See \autoref{sec:conclusion} for a discussion.)




\paragraph{Lagrangian Relaxations.}

Proving revenue approximation results via partial Lagrangians was promoted by \citet{CDW16}.  We review the basic idea here.   Consider any constrained optimization problem:
\begin{equation}
\label{eq:lr}
\begin{aligned}
   \text{maximize:} \quad & f(x) \\
   \text{subject to:} \quad & g_{\alpha}(x) \geq 0, & \forall \alpha \in A \\
	& x \in \cC
\end{aligned}
\end{equation}
Let $\mu$ be any nonnegative measure on $A$.  The Lagrangian of~\eqref{eq:lr} is $\cL(x,\lambda) \coloneqq f(x) + \int_{A} g_{\alpha}(x) \lambda(\alpha) \: \dd \mu(\alpha)$,
where $\lambda \in \R^A$ are called the Lagrange multipliers, or the dual variables.  The following fact is immediate.
\begin{fact}
   \label{fact:lagrangian}
   Let $\cC'$ be $\cC \cap \{ x \mid g_{\alpha}(x) \geq 0, \forall \alpha \in A \}$ and assume $\cC' \neq \emptyset$.  
   Then for any $\lambda \geq 0$, the optimal value of~\eqref{eq:lr} is bounded above by $\sup_{x \in \cC'} \cL(x, \lambda)$.
\end{fact}

Lagrangifying only some constraints (here $g_\alpha(x) \geq 0$) while leaving others as is (here $x \in \cC$) proves often convenient for $\cL$ to be used as a benchmark to be approximated.  In revenue maximization, often the IC and IR constraints are Lagrangified, while the feasibility constraints are left aside.


\paragraph{Equal Revenue Distributions}
Our large gap examples make use of the so-called (truncated) equal revenue distribution, which is the most commonly used distribution with unbounded gap between the surplus $\Ex{\val}$ and the optimal revenue with pricing.  
\begin{definition}[Equal Revenue Distribution]\label{def:ERD}
An equal revenue distribution truncated at $h$ is supported on $[1, h]$, with cumulative density function 
\begin{align*}
	F(\val) = \left\{ 
		\begin{array}{ll}
			1 - \frac 1 \val, & \text{for } \val \in [1, h); \\
			1, & \text{if } \val = h.
		\end{array}
		\right.
\end{align*}
\end{definition}
	It is easy to verify that $\Ex{\val} \approx \log h$, whereas any posted price extracts revenue of only~$1$.  Note that the equal revenue distribution is regular.


\section{Single-Buyer Pricing with Signals}
\label{sec:positive}

In this section we consider the problem of single-buyer single-item pricing with signals, and show our first main result: 

\begin{theorem}
   \label{thm:bounded_bsic}
   In the single-buyer pricing problem with signals, 
   suppose the distribution is jointly regular and there is no negative payment, then the revenue achievable by a mechanism with private signals is bounded by 3 times the revenue extractable by na\"ive pricing, or mechanisms with public signals.  Moreover, if the distribution is a mixture of $k$ jointly regular distributions, the gap between the two is bounded by~$3k$.
\end{theorem}

Before we give the proof in \autoref{sec:single-proof}, it is instructive to see two examples, given in \autoref{sec:examples}, which demonstrate that both conditions required by the theorem (nonnegative payment and joint regularity) are necessary for the gap between private and public signal revenues to be finitely bounded.  

	\subsection{Examples with Unbounded Value of Private Signals}
\label{sec:examples}

Our first example demonstrates that negative payments can drastically amplify the power of the private signals.

\begin{example}
	\label{ex:neg-payment}
	Let the buyer's value~$\val$ be drawn from an equal revenue distribution truncated at $H$ (see \autoref{def:ERD}).  
	For $\eps > 0$, let the signal $\sig$ be equal to $\val$ with probability~$\eps$, and $\sig = *$ otherwise.
\end{example}

Note that the distribution in this example is jointly regular. 
By the property of equal revenue distribution, with public signals, the revenue is at most $1 + \eps \ln H$.  The next theorem shows that, with private signals, the seller can get revenue $\Omega((1 - \eps) \ln \ln H)$; when $\eps$ approaches~$0$ (i.e., the ``informative'' signals vanish), the latter is arbitrarily more in comparison.  The mechanism extracting this revenue is nontrivial, a remarkable feature of it being that the ``informative'' signals are not used to extract revenue directly but to entice the buyer into truth-telling: a payment is made to the buyer if the reported value agrees with $\sig$.

\begin{theorem}
	The optimal revenue with private signals is $\Omega((1 - \eps) \ln \ln H)$. 
\end{theorem}

\begin{proof}
	We first describe the mechanism.  Recall that we use $\alloc$ and $\payment$ to denote allocation and payment rules.  Let $\alloc(\val, *)$ be $\tfrac {\ln \val}{\ln H}$, and $\payment(\val, *) = \val \alloc(\val, *) = \tfrac{\val \ln \val}{\ln H}$.
	Note that, when telling the truth, the buyer realizes zero utility under the uninformative signal $\sig = *$.  All incentives for telling the truth comes from when $\sig$ is informative and is equal to some value~$w$: whenever the reported $v \neq w$, both $\alloc(\val, w)$ and $\payment(\val, w)$ are~$0$;  otherwise, the buyer gets paid by the seller; that is, $\payment(\val, \val)$ is negative. The $\payment(\val, \val)$'s are carefully chosen to satisfy the \emph{IC} constraints: 
	let $g(z)$ be $\max_{y} \frac{\ln y}{\ln H}(z - y)$; $p_{v,v}$ is then set to be $-\frac{f(\val,*)}{f(\val,\val)} g(\val)$, where $f$ is the density function of the equal revenue distribution.


	
	
	We relegate the proofs of the following two claims to \autoref{app:examples}.
	
	\begin{claim}
		\label{cl:example-ICIR}
		The mechanism defined above is IC and ex post IR.
	\end{claim}
	
	\begin{claim}
		\label{cl:tech_g}
		For all $z \ge 1$, $g(z) \le \frac{z \ln z}{\ln H}$. If $z \ge e, g(z) \le \frac{\ln z - \ln \ln z + 1}{\ln H} \cdot z$.
	\end{claim}
	
	We now calculate the expected revenue of the mechanism.  The key is to show that the payment lost to the buyer under informative signals is offset far more by the revenue gained under the uninformative signal $\sig = *$:
	\begin{align*}
		\Rev = & \int_{1}^{H} f(v, *) \payment(\val, *) \: \dd v+ \int_{1}^{H} f(v, v) \payment(\val, \val) \: \dd v\\
		= & \int_{1}^{H} f(v,  *) (\payment(\val, *) - g(v)) \ge \int_{e}^{H} f(v,  *) \left(\frac{\ln \val}{\ln H} \cdot \val - \frac{\ln \val - \ln \ln \val + 1}{\ln H} \cdot \val \right) \: \dd v\\
	\ge & (1 - \eps) \int_{e}^{H} \frac{1}{\val} \cdot \frac{\ln \ln \val - 1}{\ln H} \: \dd v \\
	= & (1 - \eps) \cdot \frac{1}{\ln H} \cdot \left(\ln \val \ln \ln \val - 2\ln \val\right)\bigg\rvert_{e}^{H} \ge (1 - \eps) \cdot (\ln \ln H - 2),
	\end{align*}
	where the first inequality follows from \autoref{cl:tech_g}.
	
\end{proof}

Our second example shows that, in a mechanism that does not make negative payments, private signals can still blow up the revenue by an arbitrarily large factor if the supports of the value distributions under different signals are allowed to interleave in arbitrary ways, a situation that would be precluded by joint regularity.

\begin{example}
	\label{ex:weird-dist}
	Fix a large integer $\highval$.  Signal $\sig$ is uniformly distributed on $\{0, 1\}^{\highval}$.  Conditioning on a realization of~$\sig$, let the support of values, $\typespace_\sig$, be $\{3k + \sig_k : k = 1, 2, \cdots, \highval\}$, and the value conditioning on $\sig$ is drawn from the discrete equal revenue distribution supported on $\typespace_\sig$, i.e., 
\begin{align*}
	\Prx{\val = \typei | \sig} = f(t_i|\sig) \coloneqq
\begin{cases}
\frac{1}{t_i} - \frac{1}{t_{i+1}}, & i<k \\
\frac{1}{t_k}, & i=k,
\end{cases}
\end{align*}
where $t_1 < t_2 < \ldots < t_{\highval}$ are the elements of $\typespace_{\sig}$.  $f(\val | \sig)$ is $0$ for $\val \notin \typespace_\sig$.
\end{example}

The optimal revenue under public signals is easily seen to be~$1$.  Under private signals, we argue that the following mechanism is IC and extracts a revenue of $\Omega(\ln \highval)$:  Let $x(v, \sig)$ be $1$ if $f(v|\sig) > 0$, and $0$ otherwise; let $p(v, \sig)$ be $\tfrac{1}{3}v$ if $f(v|\sig) > 0$, and $0$ otherwise.

For any value~$v$ that the buyer has, telling the truth gets her utility $\mathbb{E}_{\sig} [x(v,\sig) \cdot v - p(v,\sig)] = \frac{2}{3}v$.  By reporting any other value, the probability that her report would agree with the corresponding coordinate of~$\sig$ is only $\tfrac 1 2$, which means the probability she wins the item would be at most~$\tfrac 1 2$, and her utility under deviation cannot be more than $\tfrac 1 2 v$.  This shows that the mechanism is IC.  Its revenue is clearly a third of the buyer's expected value, which is $\Omega(\ln \highval)$.




\begin{remark}
	In \autoref{ex:weird-dist}, the distribution is not jointly regular, despite its conditional value distribution resembling discretized equal revenue distributions.  Once one tries to ``pad'' the gaps in the supports to make the density nonzero everywhere, the power of the private signals disappears.  This remark is not meant to contrast discrete and continuous distributions; the point is that certain operations that are taken for granted in the absence of signals cannot be performed in correlated settings, due to the interaction among different signals.
\end{remark}

	\subsection{Proof of \autoref{thm:bounded_bsic}: Bounded Power of Private Signals under Regualrity and Non-negative Payments}
\label{sec:single-proof}

We now prove \autoref{thm:bounded_bsic}.  
The optimal revenue by a mechanism with private signal and nonnegative payments is returned by the following optimzation problem:
\begin{equation}
\label{eq:bic}
\begin{aligned}
	\max : & \int_s \int_t f(t,s)p(t,s) \: \dd t\:\dd s \\
	\text{subject to:} & \int_s f(t,s) (t\cdot x(t,s) - p(t,s)) \: \dd s \geq \int_s f(t,s)(t\cdot x(t',s) - p(t',s)) \: \dd s & \forall t,t' &&\cdots \lambda(t, t') \\
	& t\cdot x(t,s) - p(t,s) \geq 0 &\forall t,s && \cdots \mu(t,s) \\
   & p(t,s) \geq 0 & \forall t,s
\end{aligned}
\end{equation}
We will use $\Rev(x,p)$ to denote the objective value of \eqref{eq:bic} with allocation rule $x$ and payment rule $p$ and $\cF$ to denote the set of feasible $(x,p)$ pair.  Then the value of \eqref{eq:bic} is $\sup_{(x,p) \in \cF} \Rev(x,p)$.

To prove \autoref{thm:bounded_bsic}, we will attempt to bound a partial Lagrangian relaxation of~\eqref{eq:bic}.
To that end, let $\cL(x, p, \lambda, \mu)$ be the Lagrangian relaxation defined as follows.
\begin{equation}
\label{eq:lagrangian1}
\begin{aligned}
   \cL(x,p,\lambda,\mu) \coloneqq & \int_s \int_t f(t,s)p(t,s) \: \dd t \:\dd s \\
   + & \int_{t} \int_{t'} \lambda(t,t') \left( \int_{s} f(t,s) \big( t\cdot (x(t,s) - x(t',s)) - (p(t,s) - p(t',s)) \big) \: \dd s \right) \dd t' \: \dd t \\
   + & \int_{s} \int_{t} \mu(t,s) f(t,s) (t\cdot x(t,s) - p(t,s)) \: \dd t \: \dd s.
\end{aligned}
\end{equation}
Without loss of generality, we will assume $x(t,s) = p(t,s) = 0$ whenever $f(t,s) = 0$ as the partial Lagrangian remains unchanged.
   Apply Fact~\ref{fact:lagrangian}, we easily see:

\begin{lemma}
   $\sup_{(x,p) \in \cF} \Rev(x,p) \leq \inf_{\lambda,\mu \geq 0} \sup_{(x,p) \in \cF} \cL(x,p,\lambda,\mu)$.
\end{lemma}

Using Fubini and some rearrangement of \eqref{eq:lagrangian1} gives that, whenever $(x,p)$ is feasible,
\begin{align*}
   \cL(x,p,\lambda,\mu)
   = & \int_{s} \int_{t} p(t,s) \left( f(t,s) - \int_{t'} f(t,s) \lambda(t, t') \dd t' + \int_{t'} f(t',s) \lambda(t',t) \dd t' - f(t,s) \mu(t,s) \right) \dd t \: \dd s \\
   + & \int_{s} \int_{t} x(t,s) \left( tf(t,s) \mu(t,s) + \int_{t'} t f(t,s) \lambda(t, t') \: \dd t' - \int_{t'} t'f(t',s) \lambda(t',t) \: \dd t' \right) \dd t \: \dd s.
\end{align*}
Define
\[
   \lambda^*(t,t') =
   \begin{cases}
      \frac{2}{t}, & t > 0, t' \leq t \\
      0, & \text{otherwise}.
   \end{cases}
\]
Define 
\begin{align*}
g(t,s) \coloneqq f(t,s) - \int_{t'} f(t,s) \lambda^*(t,t') \: \dd t' + \int_{t'} f(t',s) \lambda^*(t', t) \: \dd t,
\end{align*} 
and let $\mu^*(t,s)$ be such that $f(t, s) \mu^*(t, s) = [g(t,s)]_+$, where $[y]_+$ denotes $\max \{0, y\}$.
  Then by the constraint $p(t,s) \geq 0$ for all $t,s$, we have
\[
   \int_{t} p(t,s) \left( f(t,s) - \int_{t'} f(t,s) \lambda(t, t') \dd t' + \int_{t'} f(t',s) \lambda(t',t) \dd t' - f(t,s) \mu^*(t,s) \right) \dd t \leq 0.
\]

Define 
\begin{align*}
	h(t,s) \coloneqq \int_{t'} tf(t,s) \lambda^*(t,t') \: \dd t' - \int_{t'} t'f(t',s) \lambda^*(t',t) \: \dd t'.
\end{align*} 
We can then bound $\cL(x, p, \lambda^*, \mu^*)$ by
\begin{align*}
   \cL(x, p, \lambda^*, \mu^*) \leq &
	\int_{s} \int_{t} x(t,s) \left( tf(t,s) \mu^*(t,s) + \int_{t'} t f(t,s) \lambda^*(t, t') \: \dd t' - \int_{t'} t'f(t',s) \lambda^*(t',t) \: \dd t' \right) \dd t \: \dd s \\
   = &
	\int_{s} \int_{t} x(t,s) \left( [tg(t,s)]_+ + h(t,s) \right) \dd t \: \dd s.
\end{align*}
\autoref{app:mix} contains the proof of the following lemma and other missing proofs of this section:
\begin{lemma}
   \label{lem:bounded1}
   The revenue of any mechanism with private signals and nonnegative prices is upper bounded by
   \begin{equation}
      \label{eq:lag2}
      \int_{s} \int_{t}\left( \left[tg(t,s) + \frac{1}{2} h(t,s) \right]_+ + [h(t,s)]_+ \right) \dd t \: \dd s.
   \end{equation}
\end{lemma}

   Our choice of $\lambda$ is independent of the signal and so we bound the inner integral of \eqref{eq:lag2} for each~$\sig$.  Let $\DRev(s)$ denote the contribution to the optimal revenue under public signals under signal~$s$.  The following lemma is the only place where we need the regularity assumption:

\begin{lemma}
   \label{lem:bounded2}
   If the distribution is jointly regular, then for all signals $s$, 
   \begin{align*}
	   \int_{t} [h(t, s)]_+ \: \dd t & \leq 2 \cdot \DRev(s), \\
	   \int_{t} \left[ t \cdot g(t, s) + \frac{1}{2} h(t, s) \right]_+ \: \dd t & \leq \DRev(s).
   \end{align*}
\end{lemma}

\begin{proof} 
   By definition, we have that $t\cdot g(t,s) = 2 \int_{t' \geq t} \frac{t f(t',s)}{t'} \: \dd t' - tf(t,s)$ and
   \begin{align*}
	   h(t,s) = 2tf(t,s) - 2 \int_{t' \geq t} f(t',s) \: \dd t' = 2f(s) \left[ tf(t | s) - 2[1 - F(t | s) \right] = 2f(s, t) \varphi(t | s).
   \end{align*}
   where $f(s) \coloneqq \int_{t'} f(t', s) \: \dd t'$ is the marginal density of~$s$, and $f(t | s)$, $F(t | s)$ and $\varphi(t | s)$ are the conditional density, conditional cumulative density and conditional virtual value of $t$ given~$s$, respectively.  By Myerson's lemma (\autoref{lem:Myerson}), $\int_t [h(t, s)]_+ \: \dd t$ is bounded by $2\DRev(s)$.


   For the second part of the statement, we have
   \[
	   \int_{t} \left[ t \cdot g(t,s) + \frac{1}{2}h(t,s), \right]_+ \: \dd t
      = \int_{t} t \cdot \left[ 2 \int_{t' \geq t} \frac{f(t',s)}{t'} \: \dd t' - \int_{t' \geq t} \frac{f(t',s)}{t} \: \dd t' \right]_+ \: \dd t.
   \]
   We first show that there exists $0 \leq l \leq u$ such that the integrand is positive only on an interval $[l, u]$.  Let $\psi_s(t)$ be  $2 \int_{t' \geq t} \frac{f(t',s)}{t'} \: \dd t' - \int_{t' \geq t} \frac{f(t',s)}{t} \: \dd t'$.  Then
   \begin{align*}
      \frac{\dd \psi_s}{\dd t}(t) = & -2\frac{f(t,s)}{t} + \int_{t' \geq t} \frac{f(t',s)}{t^2} \: \dd t' + \frac{f(t,s)}{t} \\
      = & -\frac{f(t,s)}{t^2} \left( t - \frac{1 - F(t|s)}{f(t|s)} \right) \\
      = & -\frac{f(t,s)}{t^2} \varphi(t|s).
   \end{align*}
   By joint regularity, $\varphi(t|s)$ is nondecreasing so $\psi_s(t)$ has at most two roots.  In particular, there exists $0 \leq l \leq u$ such that $\psi_s(t) \geq 0$ if and only if $l \leq t \leq u$.
   For any $r \in \R_{\geq 0}$ we have 
   \begin{align*}
      \int_{t \leq r} t \cdot \psi_s(t) \: \dd t = & \int_{t \leq r} 2 \int_{t' \geq t} \frac{t f(t',s)}{t'} \: \dd t' \: \dd t - \int_{t \leq r} \int_{t' \geq t} f(t',s) \: \dd t' \: \dd t \\
      = & \int_{t'} \int_{t \leq \min\{r, t'\}} \frac{2t f(t',s)}{t'} \: \dd t \: \dd t' - \int_{t'} \int_{t \leq \min\{r,t\}} f(t',s) \: \dd t \: \dd t' \\
      = & \int_{t' \leq r} \left( \int_{t \leq t'} \frac{2tf(t',s)}{t'} \: \dd t - t'f(t',s) \right) \dd t'
      + \int_{t' > r} \left(  \int_{t \leq r} \frac{2tf(t',s)}{t'} \dd t - rf(t',s) \right) \dd t' \\
      = & \int_{t' > r} \left( \frac{r^2}{t'} - r \right) f(t',s) \: \dd t' \leq0.
   \end{align*}
   On the other hand, for any $r \in \mathbb R_{\geq 0}$,
   \begin{align*}
	-\int_{t \leq r} t \cdot \psi_s(t) \: \dd t = - \int_{t' > r} \left( \frac{r^2}{t'} - r \right) f(t',s) \: \dd t'
	\leq r \int_{t' \geq r} f(t',s) \: \dd t' \leq \DRev(s).
\end{align*}

Thus, $\int_{t} [t\cdot g(t,s) + \frac{1}{2} h(t,s)]_+ \: \dd t = \int_{l \leq t \leq u} t \cdot \psi_s(t) \: \dd t = \int_{t \leq u} t \cdot \psi_s(t) \: \dd t - \int_{t \leq l} t \cdot \psi_s(t) \: \dd t \leq \DRev(s)$.
\end{proof}

Putting together what we have showed: by \autoref{lem:bounded1}, the revenue of any mechanism with private signal and nonnegative payments is bounded above by \eqref{eq:lag2}.
   Assuming joint regularity and applying \autoref{lem:bounded2},
   we have that the inner integral of~\eqref{eq:lag2} is bounded above by $3 \DRev(s)$.
   Hence, the revenue of any BSIC mechanism is bounded above by $3 \int_s \DRev(s) \: \dd s = 3\DRev$.

   When the distribution is a mixture of $k$ jointly regular distributions, all arguments remain true except that the bounds in \autoref{lem:bounded2} degrade by a factor of~$k$.  We relegate the details to \autoref{app:mix}.

\section{Multi-bidder Auctions}
\label{sec:lookahead}

In this section we consider the revenues of DSIC and BIC multi-bidder auctions with correlated values.  Our examples in \autoref{sec:examples} can be easily translated to multi-bidder distributions, showing:

\begin{corollary}
	\label{cor:lookahead-examples}
In multi-bidder auctions with correlated values, if the auctioneer is allowed to make payments to the bidders, or if the value distribution is not jointly regular, the revenue of the optimal BIC mechanism can be more than that of any DSIC mechanism by an arbitrarily large factor.
\end{corollary}

For completeness, we give a proof for the corollary in \autoref{app:lookahead}.  The main result of this section is that the bound we gave in \autoref{thm:bounded_bsic} for the single-buyer problem can be extended to multi-bidder auctions, with the loss of another constant factor.  It is tempting to use \citet{R01}'s lookahead auction (see \autoref{def:DSLA}) which reduces the multi-bidder problem to extracting revenue from the highest bidder.  However, \citeauthor{R01}'s proof does not directly apply to BIC mechanisms due to subtle reasons we discuss in \autoref{sec:conclusion}.  Instead, we are able to carry on a ``lookahead type of argument'' in the partial Lagrangian under no negative payment and regularity: almost coincidentally, the dual variables we constructed in the proof of \autoref{thm:bounded_bsic} have a property needed for such an argument (\autoref{cl:dual-bound}). 

\begin{theorem}
\label{thm:lookahead-nonneg}
In a multi-bidder auction with jointly regular distribution, the revenue of the lookahead auction is at least $\tfrac 1 5$ that of the optimal BIC mechanism not using negative payments.  Furthermore, if the the distribution is a mixture of $k$ jointly regular distributions, the bound is $\tfrac 1 {3k + 2}$.
\end{theorem}

\begin{proof}
	Following an approach similar to \autoref{sec:single-proof}, we show the following lemma, whose proof and other missing proofs from this section can be found in \autoref{app:lookahead}.

	\begin{lemma}
		\label{lem:lookahead-lagrange}
		The revenue of the optimal BIC, ex post IR mechanism is at most 
		\begin{align*}
			\max_{\allocs \geq \mathbf{0} } & \int_{\types} f(\types) \sum_i \alloci(\types) \left\{ [g_i(\types) \typei]_+ + h_i(\types) \right\} \: \dd \types, \\
\text{s.t.} \quad & \sum_i \alloci(\types) \leq 1, \qquad \forall \types,
		\end{align*}
		where
		\begin{align*}
			g_i(\types) & =  - f(\types)  + 2\int_{\typei' > \typei} \frac{f(\typei', \typesmi)}{\typei'} \: \dd \typei',\\
			h_i(\types) & =  2f(\typesmi)\left[ \typei f(\typei | \typesmi) - (1 - F(\typei | \typesmi) \right].  
		\end{align*}
	\end{lemma}

\begin{claim}
	\label{cl:dual-bound}
	For any $\types$ and~$i$, $[g_i(\types)\typei]_+ + h_i(\types) \leq 2\typei f(\types)$.
\end{claim}

Let $U_i \subseteq \typespace$ denote the set of type profiles in which bidder~$i$ is considered the potential winner by the lookahead auction.  Recall that $\types \in U_i$ implies that $\typei$ is a highest type in~$\types$.  Let $\secmax(\types)$ denote the second highest type in vector~$\types$.  
We further partition $L(\mu^*, \lambda^*)$:
\begin{align}
	L(\mu^*, \lambda^*) &
   \leq \max_{\allocs \geq \mathbf{0}} \sum_i \left\{ \int_{\types \in U_i} \sum_{i'} \alloci[i'](\types) ([\typei[i'] g_{i'}(\types)]_+ + h_{i'}(\types)) \: \dd \types \right\} \nonumber \\
   & \leq \max_{\allocs \geq \mathbf{0}} \sum_i \left\{ \int_{\types \in U_i} \sum_{i' \neq i} \alloci[i'](\types) ([\typei[i'] g_{i'}(\types)]_+ + h_{i'}(\types)) \: \dd \types + \int_{\types \in U_i} \alloci(\types)( [\typei g_i(\types)]_+ + h_i(\types)) \: \dd \types \right\} \nonumber \\
       &\leq \sum_i \int_{\types \in U_i} 2 \secmax(\types) f(\types) \: \dd \types + \sum_i \int_{\types \in U_i} \alloci(\types) ([g_i(\types)]_+ + h_i(\types) ) \: \dd \types  \nonumber \\
       & = 2 \int_{\types} \secmax(\types) f(\types) \: \dd \types + \sum_i \int_{\types \in U_i} \alloci(\types) ([g_i(\types)]_+ + h_i(\types)) \: \dd \types \label{eq:lookahead-two-term} \\
	\text{s.t.} \quad & 0 \leq \alloci(\types) \leq 1, \qquad \forall \types \in U_i. \nonumber
\end{align}
where the inequality follows from \autoref{cl:dual-bound}, $\sum_{i'\neq i} \alloci[i'] \leq 1$ and $\typei[i'] \leq \secmax(\types)$ for any $\types \in U_i$.  The first term of \eqref{eq:lookahead-two-term} is twice the revenue of the second price auction.  With slight modification of the argument in the proof of \autoref{thm:bounded_bsic} (see \autoref{lem:lookahead-bounded2} in \autoref{app:lookahead}), one can show that the second term of \eqref{eq:lookahead-two-term} is upper bounded by $3k$ times the revenue of the lookahead auction when the distribution is a mixture of $k$ jointly regular distributions.  This completes the proof of \autoref{thm:lookahead-nonneg}.
\end{proof}

\section{Conclusions and Discussion}
\label{sec:conclusion}

In this work we give a fairly complete characterization, up to constant multiplicative factors, of revenue differences between disclosed and private side information when pricing for a single buyer, and of revenues differences between DSIC and BIC auctions for multiple bidders with correlated values.  We find that regularity in the value distribution and the absence of negative payments are necessary and sufficient for these differences to be relatively small.

In connecting the results for these two settings, we used an argument that is reminiscent of \citet{R01}'s proof for his lookahead auction (\autoref{def:DSLA}).  However, our proof relies on the two mentioned conditions and our particular choice of dual variables.  \citeauthor{R01}'s original proof cannot be straightforwardly carried over for BIC auctions, and we briefly explain the reason here.

One property trivially true for DSIC auctions is convenient for showing the approximate optimality of the lookahead auction: given an optimal DSIC auction, if one were to zero out all the allocations and payments for all bidders except for the highest, the remaining auction is still DSIC.  This is not true for BIC auctions. A bidder holding a certain value is sometimes the highest bidder, and sometimes not; her expected utility is affected by her allocations and payments when her value is not the highest. Removing these allocations and payments may invalidate the IC constraints.  Therefore, by restricting an auction to only allocating to the highest bidders, it is unclear whether it is able to extract as much revenue as an unrestricted auction from the highest bidders.  \footnote{This also calls to mind the comparison between $\beta$-exclusive and $\beta$-adjusted revenues considered by \citet{Yao15} for multi-item auctions with bidders with additive values.  It is also unclear whether \citeauthor{Yao15}'s argument can be applied here.}

 We therefore leave open the following interesting question: for multiple bidders with correlated values, can a BIC auction that allocates only to the highest bidders extract a constant fraction of the optimal BIC auction's revenue?

\appendix
\section{Proofs from \autoref{sec:examples}}
\label{app:examples}

%
%
%

\begin{proof}[Proof of \autoref{cl:example-ICIR}]
	Showing ex post IR is straightforward: for the informative signals, it is easy to see $\payment(\val, w) \le 0 \leq \alloc(\val, w) \cdot \val$ for all $v,w$;  for $s=*$, $\payment(\val, *)$ equals $\alloc(\val, *) \cdot v$ by definition.  We now show that mechanism is also IC. By reporting her true value, a buyer with value~$\val$ has expected utility
	\[
		\frac{1}{f(\val, *) + f(\val, \val)} \big( f(v, *) \cdot 0 - f(v, v) \cdot \payment(v,v) \big) = \frac{f(v, *)}{f(v, *) + f(v, v)} \cdot g(v).
	\]
	By telling her value as~$w$, the buyer would have expected utility
	\begin{align*}
		\frac{1}{f(v, *) + f(v, v)} \big( f(v, *) \cdot (v \cdot \alloc(\val, *) - \payment(\val, *)) + f(v, v) \cdot 0 \big) & = \frac{f(v, *)}{f(v, *) + f(v, v)} \cdot \frac{\ln w}{\ln H}(v-w) \\
																      & \le \frac{f(v, *)}{f(v, *) + f(v, v)} \cdot g(v),
	\end{align*}
	where the inequality holds by the definition of $g(z)$.
\end{proof}

\begin{proof}[Proof of \autoref{cl:tech_g}]
	
	Let $h(y) = \ln y\cdot(z-y)$.
	Observe that when $y > z$, $h(y) < 0 \le z \ln z$ and when $y \le z$, $h(y) \le z \ln z$, the first part of the claim is trivial.
	Now suppose $z \ge e$.
	Then $h'(y) = \frac{1}{y}(z-y)-\ln y$. Hence the maximum of $g(z)$ is achieved when $y(\ln y + 1) = z$. 
	This implies $\ln y \le \ln z - \ln \ln z + 1$, otherwise $y(\ln y+1) > \frac{ez}{\ln z}(\ln z-\ln \ln z + 2) \ge z$.
	Therefore,
	$g(z) = \frac{\ln y}{\ln H}(z-y) \leq \frac{\ln z - \ln \ln z + 1}{\ln H} \cdot z$ when $z \ge e$.

\end{proof}

\section{Proofs from \autoref{sec:single-proof}}
\label{app:mix}

\begin{proof}[Proof of \autoref{lem:bounded1}]
   The optimal revenue is bounded above by
   \begin{align*}
	   \sup_{(x,p) \in \cF} \cL(x, p, \lambda^*, \mu^*) & \leq
   \sup_{(x,p) \in \cF} \int_{s} \int_{t} x(t,s) \left( [tg(t,s)]_+ + h(t,s) \right) \dd t \: \dd s \\
      & =
      \int_{s} \int_{t} \left[[tg(t,s)]_+ + h(t,s) \right]_+ \dd t \: \dd s.
    \end{align*}
    Let $\typespace_+$ denote $\{(t, s): g(t, s) \geq 0\}$ and let $1_{\typespace_+}$ be its indicator function, then
    \begin{align*}
    & \int_{s} \int_{t} \left[[tg(t,s)]_+ + h(t,s) \right]_+ \dd t \: \dd s \\
	    =&
	    \int_s \int_t [tg(t,s) + h(t, s)]_+ 1_{\typespace_+} \dd t \: \dd s + \int_s \int_t [h(t, s)]_+ 1_{\overline{\typespace_+}} \dd t \: \dd s \\
	    \leq &
	    \int_s \int_t \left[ t g(t, s) + \frac 1 2 h(t, s) \right]_+ 1_{\typespace_+} \: \dd t \: \dd s + \frac 1 2 \int_s \int_t [h(t, s)]_+ 1_{\typespace_+} \: \dd t \: \dd s + \int_s \int_t [h(t,s)]_+ 1_{\overline{\typespace_+}} \: \dd t \: \dd s \\
	    \leq & 
	    \int_{s} \int_{t} \left( \left[ tg(t,s) + \frac{1}{2} h(t,s) \right]_+ + [h(t,s)]_+  \right) \: \dd t \: \dd s. \qedhere
   \end{align*}
\end{proof}

\begin{proof}[Proof for Mixture of $k$ joingly regular distributions]
	Adapting the same flow variables $\lambda(\cdot, \cdot)$ as the proof for \autoref{thm:bounded_bsic}, we only need to prove the following variation of \autoref{lem:bounded2}: if $f(\cdot, s)$ is a mixture of $k$ regular distributions, that is, $f(\type, \sig) = \sum_{i = 1}^k \alpha_i f_i(\type, \sig)$, where each $f_i$ is the density of a jointly regular distribution, and $\sum_i \alpha_i = 1$, with $\alpha_i \geq 0, \forall i$, then	
\begin{align*}
	   \int_{t} [h(t, s)]_+ \: \dd t & \leq 2k \cdot \DRev(s), \\
	\int_{t} \left[ t \cdot g(t, s) + \frac{1}{2} h(t, s) \right]_+ \: \dd t & \leq k \cdot \DRev(s).
\end{align*}

As before, we have
\begin{align*}
	h(t,s) & = 2tf(t,s) - 2 \int_{t' \geq t} f(t',s) \: \dd t', \\
\int_{t} \left[ t \cdot g(t,s) + \frac{1}{2}h(t,s), \right]_+ \: \dd t
& = \int_{t} t \cdot \left[ 2 \int_{t' \geq t} \frac{f(t',s)}{t'} \: \dd t' - \int_{t' \geq t} \frac{f(t',s)}{t} \: \dd t' \right]_+ \: \dd t.
\end{align*}
Using the fact that $[x+y]_+ \leq [x]_+ + [y]_+$, we have
\begin{align*}
	\int_{\type} [h(\type, \sig)]_+ \: \dd \type & \leq \sum_{i \in [k]} 2 \alpha_i \int_{\type} f(\type, \sig) [\virt_{\sig}^i (\type)]_+ \: \dd \type, \\
\int_{t} \left[ t \cdot g(t,s) + \frac{1}{2}h(t,s), \right]_+ \: \dd t
& \leq \sum_{i \in [k]} \alpha_i \int_{t} t \cdot \left[ 2 \int_{t' \geq t} \frac{f_i(t',s)}{t'} \: \dd t' - \int_{t' \geq t} \frac{f_i(t',s)}{t} \: \dd t' \right]_+ \: \dd t,
\end{align*}
where $\virt_{\sig}^i$ denotes the conditional virtual value of the $i$-th regular distribution given signal~$\sig$.  Now apply the argument in \autoref{sec:single-proof} to get
\begin{align*}
	\int_{\type} f(\type, \sig) [\virt_{\sig}^i(\type)]_+ \: \dd \type & \leq 2 \DRev_i(\sig),\\
	\int_{t} t \cdot \left[ 2 \int_{t' \geq t} \frac{f_i(t',s)}{t'} \: \dd t' - \int_{t' \geq t} \frac{f_i(t',s)}{t} \: \dd t' \right]_+ \: \dd t & \leq \DRev_i(s),
\end{align*}
where $\DRev_i(s)$ denotes the part of the optimal DSIC revenue for distribution $f_i(\cdot, s)$ from signal~$\sig$.
Observe that $\alpha_i \DRev_i(s) \leq \DRev(s)$ for all $i$, we conclude the proof.
\end{proof}

\section{Proofs from \autoref{sec:lookahead}}
\label{app:lookahead}

\begin{proof}[Proof of \autoref{cor:lookahead-examples}]
	Take either example from \autoref{sec:examples}, we construct a two-bidder auction where the revenue of the optimal BIC mechanism is at least the revenue in the example under private signals, and revenue of the optimal DSIC mechanism is arbitrarily close to the optimal revenue under disclosed signals.  In both examples, the signal space $\sigspace$ is finite, and we can number the signals by integers so that $\sigspace = \{1, 2, \cdots, |\sigspace|\}$.  For every $(\val, \sig)$ in the single-buyer example, let bidder~$1$'s value be $\val$, and bidder~$2$'s value be $\sig \cdot \eps / |\sigspace|$ for an arbitrarily small $\eps > 0$.  Take any single-buyer selling mechanism with private signals, use its allocation and payment rules for bidder~$1$ while taking bidder~$2$'s value as the signal (while never allocating or charging anything to bidder~$2$), then the resulting auction is BIC and extracts the same revenue as in the single-buyer example.  On the other hand, if pricing under disclosed signals extracts at most revenue $R$ from the single buyer, then any DSIC auction for the two-bidder setting can extract at most $R$ from bidder~$1$, and at most $\eps$ from bidder~$2$ (because of individual rationality).  The corollary follows by taking $\eps$ sufficiently small.

\end{proof}

\begin{proof}[Proof of \autoref{lem:lookahead-lagrange}]
We consider the linear program that returns the revenue of the optimal BIC multi-bidder auction without negative payments.  

\begin{align*}
	\max_{\allocs, \payments} & \int_{\types} f(\types) \sum_i \paymenti(\types) \: \dd \types &  \\
   \text{s.t.} & \int_{\typesmi} f(\typei, \typesmi) [\typei \cdot \alloci(\typei, \typesmi) - \paymenti(\typei,\typesmi)] \: \dd \typesmi \geq \int_{\typesmi} f(\typei,\typesmi) [\typei \cdot \alloci(\typei',\typesmi) - \paymenti(\typei',\typesmi)] \: \dd \typesmi \quad \forall i, \typei, \typei' \\
   & \typei \cdot \alloci(\types) - \paymenti(\types) \geq 0, \quad \forall i, \types \\
	& \paymenti(\types) \geq 0, \quad \forall i, \types \\
   &  \sum_i \alloci(\types) \leq 1, \quad \forall \types \\
	  & \alloci(\types) \geq 0, \quad \forall i, \types
\end{align*}

 This revenue is upper bounded by the following partial Lagrangian, for any $\mathbf{\mu, \lambda} \geq 0$:
\begin{align*}
	L(\bfmu, \bflambda) \coloneqq & 
	\max_{\allocs, \payments} \int_{\types} f(\types) \sum_i \paymenti(\types) \: \dd\types + \int_{\types} \sum_i \mu_i(\types) [\typei \cdot \alloci(\types) - \paymenti(\types)] \: \dd \types \\
   & + \sum_i \int_{\typei} \int_{\typei'} \lambda_i(\typei, \typei') \left\{ \int_{\typesmi} f(\typei, \typesmi) [\typei \cdot \alloci(\typei, \typesmi) - \paymenti(\typei,\typesmi) - \typei \cdot \alloci(\typei',\typesmi) + \paymenti(\typei',\typesmi)] \dd \typesmi \right\} \: \dd\typei \dd\typei' \\
   = & \max_{\allocs, \payments} \sum_i \int_{\types} \paymenti(\types) \left\{f(\types) - \mu_i(\types) + \int_{\typei'} \left[\lambda_i(\typei', \typei) f(\typei', \typesmi) - \lambda_i(\typei, \typei') f(\typei, \typesmi) \right] \dd \typei' \right\} \dd \types\\
       & + \sum_i \int_{\types} \alloci(\types) \left\{ \mu_i(\types) \typei + \int_{\typei'} \left[ \typei \cdot \lambda_i(\typei, \typei') f(\typei, \typesmi) - \typei' \cdot \lambda_i(\typei', \typei) f(\typei', \typesmi) \right] \dd \typei'
\right\} \dd \types \\
   \text{s.t.} \quad & \paymenti(\types) \geq 0, \qquad \forall i, \types \\
& \sum_i \alloci(\types) \leq 1, \qquad \forall \types \\
	  & \alloci(\types) \geq 0, \qquad  \forall i, \types
\end{align*}

Similarly to the proof of \autoref{thm:bounded_bsic}, we try the following dual variables:
\begin{align*}
\lambda_i^*(\typei, \typei') = \left\{ 
	\begin{array}{ll}
      \frac{2}{\typei}, & t > 0, t' \leq t; \\
      0, & \text{otherwise}.
\end{array}
\right.
\end{align*}
Again, define
\begin{align*}
	g_i(\types) \coloneqq f(\types) - \int_{\typei'} f(\types) \lambda^*_i(\typei, \typei') \: \dd t' + \int_{\typei'} f(\typei', \typesmi) \lambda^*_i(\typei', \typei) \: \dd \typei' = - f(\types)  + 2\int_{\typei' > \typei} \frac{f(\typei', \typesmi)}{\typei'} \: \dd \typei',\\
	h_i(\types) = \int_{\typei'} \typei \cdot \lambda^*_i(\typei, \typei') f(\types) \: \dd \typei' - \int_{\typei'} \typei' \cdot \lambda_i^*(\typei', \typei) f(\typei', \typesmi) \: \dd \typei' = 2f(\typesmi)\left[ \typei f(\typei | \typesmi) - (1 - F(\typei | \typesmi) \right].
\end{align*}

Let $\mu_i(\types)$ be $[g_i(\types)]_+$.  Then 
\begin{align*}
	L(\mu^*, \lambda^*) \leq & \max_{\allocs \geq \mathbf{0} } \int_{\types} f(\types) \sum_i \alloci(\types) \left\{ [g_i(\types) \typei]_+ + h_i(\types) \right\} \: \dd \types, \\
\text{s.t.} \quad & \sum_i \alloci(\types) \leq 1, \qquad \forall \types.
\end{align*}
\end{proof}

\begin{proof}[Proof of \autoref{cl:dual-bound}]
The claim is easily seen by a case analysis.  For $g_i(\types) \leq 0$, $[g_i(\types)\typei]_+ + h_i(\types) = h_i(\types) \leq 2f(\types) \typei$, as $1 - F(\typei | \typesmi) \geq 0$.  For $g_i(\types) \geq 0$, we have
\begin{align*}
	[g_i(\types) \typei]_+ + h_i(\types) = \typei f(\types) + 2 \left[\int_{\typei' > \typei} \frac{\typei f(\typei', \typesmi)}{\typei'} \: \dd \typei' - \int_{\typei' > \typei} \frac{\typei' f(\typei', \typesmi)}{\typei'} \: \dd \typei' \right] \leq \typei f(\types).
\end{align*}
\end{proof}

\begin{lemma}
	\label{lem:lookahead-bounded2}
Let $\Rev_i$ be the revenue extracted from bidder~$i$ by the lookahead auction, then 
\begin{align*}
	\max_{\mathbf 0 \leq \alloci(\types) \leq \mathbf 1} \alloci(\types) \int_{\types \in U_i} ([\typei g_i(\types) ]_+ + h_i(\types) ) \: \dd \types,
\end{align*}
is bounded by $3\Rev_i$, where $g_i$ and $h_i$ are defined as in~\eqref{eq:lookahead-two-term},
\end{lemma}

\begin{proof}
	By the same argument as in the proof of \autoref{lem:bounded1}, we have
	\begin{align*}
		\max_{\mathbf 0 \leq \alloci(\types) \leq \mathbf 1} \alloci(\types) \int_{\types \in U_i} ([\typei g_i(\types) ]_+ + h_i(\types) ) \: \dd \types  \leq \int_{\types \in U_i} [\typei g_i(\types) + \frac 1 2 h_i(\types) ]_+ + [h_i(\types) ]_+ \: \dd \types.
	\end{align*}
	By the same calculation as in the proof of \autoref{lem:bounded2}, $h_i(\types)$ is $2 f(\types) \virt_i(\typei | \typesmi)$, where $\virt_i(\typei | \typesmi)$ is the virtual value of~$\typei$ in the distribution conditioning on $\typesmi$.  By joint regularity, $\typedist(\typei | \typesmi)$ is regular, and so $\typedist(\typei | \typesmi, \typei > \max_{i'\neq i} \typei[i'])$ is also regular.  Therefore 
	\begin{align*}
		\int_{\types \in U_i} [h_i(\types)]_+ \: \dd \types = 2 \Rev_i. 
	\end{align*}
	
	As for the term $[\typei g_i(\types) + \tfrac 1 2 h_i(\types)]_+$, from the calculation in \autoref{lem:bounded2}, we know that, fixing any $\typesmi$, $\typei g_i(\types)$ is nonnegative only on an interval $[l, u]$, and 
	\begin{align}
		\int_{\typei \leq u: (\typei, \typesmi) \in U_i} \left( \typei g_i(\types) + \frac 1 2 h_i(\types) \right) \: \dd \typei & \leq 0, \nonumber \\
		- \int_{\typei \leq l: (\typei, \typesmi) \in U_i} \left( \typei g_i(\types) + \frac 1 2 h_i(\types) \right) \: \dd \typei & \leq l \int_{\typei \geq l: (\typei, \typesmi) \in U_i} f(\typei, \typesmi) \: \dd \typei \nonumber \\
   & \leq \max(l, \max_{i' \neq i} \typei[i'] ) \int_{\typei \geq l: (\typei, \typesmi) \in U_i} f(\typei, \typesmi) \: \dd \typei,
		\label{eq:psi-lookahead}
	\end{align}
	where the RHS of \eqref{eq:psi-lookahead} is upper bounded by the revenue of the lookahead auction when the bidders other than~$i$ bid~$\typesmi$.  Therefore,
	\begin{align*}
		\int_{\types \in U_i} [\typei g_i(\types) + \frac 1 2 h_i(\types)]_+ \: \dd \types \leq \Rev_i.
	\end{align*}
\end{proof}

\bibliographystyle{apalike}
\bibliography{bibs}

\end{document}